\documentclass{article}
\usepackage{fullpage}
\usepackage[utf8]{inputenc}
\usepackage[all=normal,bibliography=tight]{savetrees}

\usepackage{amstext,amsfonts,amsthm,amsmath,amssymb}
\usepackage{graphicx,tikz}
\usepackage[colorlinks=true,citecolor=red]{hyperref}
\usepackage[capitalize]{cleveref}
\usepackage{enumerate}   
\usepackage[linesnumbered,ruled,vlined]{algorithm2e}

\newtheorem{lemma}{Lemma}[section]
\newtheorem{theorem}[lemma]{Theorem}
\newtheorem{claim}[lemma]{Claim}
\newtheorem{observation}[lemma]{Observation}
\newtheorem{proposition}[lemma]{Proposition}
\newtheorem{corollary}[lemma]{Corollary}
\newtheorem{question}[lemma]{Problem}

\theoremstyle{definition}
\newtheorem{definition}[lemma]{Definition}

\newcommand{\defproblem}[3]{
 \vspace{2mm}
\noindent\fbox{
 \begin{minipage}{0.96\textwidth}
 \begin{tabular*}{\textwidth}{@{\extracolsep{\fill}}lr} #1 & \\ \end{tabular*}
 {\textbf{Input:}} #2 \\
 {\textbf{Question:}} #3
 \end{minipage}
 }
 \vspace{2mm}
}
\newcommand{\M}{M}
\newcommand{\colres}{\texttt{p}}
\def\r{\texttt{r}}

\def\auxi{\texttt{aux}}
\def\pos{\texttt{pos}}
\def\poss{\mathrm{pos}}
\def\process{{\texttt{process} }}

\DeclareMathOperator{\tw}{tw}

\DeclareMathOperator{\tpw}{tpw}

\newcommand{\N}{\mathbb{N}}

\title{List Colouring Trees in Logarithmic Space}
\author{Hans L. Bodlaender\footnote{Utrecht University, \texttt{h.l.bodlaender@uu.nl}.} \quad Carla Groenland\footnote{Utrecht University, \texttt{c.e.groenland@uu.nl}. Supported by the project CRACKNP that has received funding from the European Research Council (grant agreement No 853234).} \quad Hugo Jacob\footnote{ENS Paris-Saclay, hjacob@ens-paris-saclay.fr.}}
\date{\today}

\begin{document}
\maketitle 
\begin{abstract}
    We show that \textsc{List Colouring} can be solved on $n$-vertex trees by a deterministic Turing machine using $O(\log n)$ bits on the worktape. Given an $n$-vertex graph $G=(V,E)$ and a list $L(v)\subseteq\{1,\dots,n\}$ of available colours for each $v\in V$, a list colouring for $G$ is a proper colouring $c$ such that $c(v)\in L(v)$ for all $v$.
\end{abstract}

\section{Introduction}
Various applications can be modelled as an instance of \textsc{List Colouring}, e.g.,
the vertices may correspond to communication units, with lists giving the possible frequencies or channels that a vertex may choose from as colours and edges showing which units would interfere if they are assigned the same colour \cite{GPT,RBAB}. 

Given a graph $G=(V,E)$ and given a list $L(v)$ of colours for each vertex $v\in V$, an \textit{$L$-colouring} $c$ is a
proper colouring (that is, $c(u)\neq c(v)$ when $uv\in E$) mapping every vertex $v$ to a colour in the list $L(v)$. This gives rise to the following computational problem.

\defproblem{\sc List Colouring}{A graph $G=(V,E)$ with a list $L(v)\subseteq \{1,\dots,n\}$ of available colours for each $v\in V$.}{Is there an $L$-colouring for $G$?}\\

\textsc{List Colouring} is computationally hard. It is NP-complete on cographs \cite{JansenScheffler97} and on planar bipartite graphs, even when all lists are of size at most 3 \cite{GravierKoblerKubiak02}.
The problem remains hard when parameterised by `tree-like' width measures: it was first shown to be W$[1]$-hard parameterised by treewidth in 2011 by \cite{FellowsFLRSST11} and recently shown to be XNLP-hard implying W[$t$]-hardness for all $t$ by \cite{XNLP-comp}.
On the other hand, on $n$-vertex trees the problem can be solved in time linear in $n$ (using hashing)  \cite{JansenScheffler97}, but this algorithm may use $\Omega(n)$ space.

In this paper, we study the auxiliary space requirements of \textsc{List Colouring} on trees in terms of the number of vertices $n$ of the tree. We assume that the vertices of $T$ have been numbered $1,\dots,n$, which gives a natural order on them, and that, given vertices $v,v'$ in $T$ and $i,i'\in \{1,\dots,n\}$, it can be checked in $O(\log n)$ space whether the $i$th colour in $L(v)$ equals the $i'$th colour in $L(v')$. As is usual for the complexity class L (logspace), we measure the space requirements in terms of the number of bits on the work tape of a deterministic Turing machine, where the description of the tree and the lists are written on a (read-only) input tape. In particular, the number of bits on the input tape is allowed to be much larger. 

Since $n$-vertex trees have pathwidth $O(\log n)$, our problem can be solved non-deterministically using $O(\log^2 n)$ bits on the work tape (see Proposition \ref{prop:max_degree}). However, doing better than this is surprisingly challenging, even in the non-deterministic case!
Our main result is as follows. 
\begin{theorem}
\label{thm:main}
{\sc List Colouring} for trees is in \textup{L}.
\end{theorem}
Our initial interest in the space complexity of \textsc{List Colouring} on trees arose from a recent result showing that  \textsc{List Colouring} parameterised by pathwidth is XNLP-complete \cite{XNLP-comp}.
XNLP is the class of problems on an input of size $n$ with parameter $k$, which can be solved by a non-deterministic Turing machine in $f(k)n^{O(1)}$ time and $f(k)\log n$ space for some computable $f$. Since the treewidth of a graph is upper bounded by the pathwidth, \textsc{List Colouring} is also XNLP-hard parameterised by treewidth. This is conjectured\footnote{The conjecture (see \cite[Conjecture 2.1]{PW18} or \cite[Conjecture 5.1]{XNLP-comp}) states that there is no deterministic algorithm for an XNLP-hard problem that runs in XP time and FPT space. If for each $k$, there exists a constant $f(k)$ such that \textsc{List Colouring} can be solved in space $f(k)\log n$ on $n$-vertex graphs of treewidth $k$, then this would in particular yield an algorithm running in $n^{f(k)}$ time and $f(k)n^{O(1)}$ space.} to imply that there is a constant $k^*$  for which any deterministic Turing machine needs $\omega(\log n)$ space in order to solve \textsc{List Colouring} for $n$-vertex graphs of treewidth $k^*$; this work shows that $k^*>1$. 
It seems likely that  \textsc{List Colouring} parameterised by treewidth is not in XNLP, and we conjecture that it is complete for a parameterised analogue of NAuxPDA (also known as SAC) from \cite{AllenderCLPT14,PW18} .
Considering such classes which (also) have space requirements (complexity classes such as XL, XNL and XNLP \cite{XNLP-comp,BodlaenderGS21,ChenF03,ElberfeldST15})  has proven successful in classifying the complexity of parameterised problems which are not known to be complete for any classes that only consider time requirements. Since some of such classes are very naturally modelled by instances of \textsc{List Colouring}, we believe the complexity class of \textsc{List Colouring} on trees could be of theoretical interest as well.

Another motivation for studying space requirements comes from practice, 
since memory can be much more of a bottleneck than
processing time (e.g. for dynamic programming approaches). This motivates the development of techniques to reduce the space complexity. Although many techniques have been established to provide algorithms which are efficient with respect to time, fewer techniques are known to improve the space complexity. Notable exceptions include the logspace analogue of Bodlaender's and Courcelle's theorem \cite{ElberfeldJT10} which allows one to check any monadic second-order formula on graphs of bounded treewidth in logspace (which in particular allows one to test membership in any minor-closed family) and Reingold's \cite{Reingold} work on undirected connectivity. Reachability and isomorphism questions have also been well-studied on restricted graph classes, e.g.  \cite{Datta,direachgenus,Lindelltreeisom}. 
Another interesting piece of related work \cite{EgriHLR14} shows that for each graph $H$, \textsc{List $H$-Colouring} is either in L or NL-hard. We remark that \textsc{List $H$-Colouring} on trees
(for fixed $H$)
is easily seen to be solvable in logspace using an analogue of Proposition \ref{prop:max_degree} or the logspace Courcelle's theorem \cite{ElberfeldJT10}. The difficulty in our case comes from the fact that the sizes of the lists are unbounded.

We also generalise our algorithm to graphs of bounded tree-partition-width (also called strong treewidth).
\begin{corollary}
\label{cor:listcol_tpw}
There is a deterministic $O(k\log k \log n)$ space algorithm for \textsc{List Colouring} on $n$-vertex graphs with a given tree-partition of width $k$.
\end{corollary}
The algorithm of Corollary \ref{cor:listcol_tpw} does not run in FPT time.
We also include a simple proof that \textsc{List Colouring} is W[1]-hard when parameterised by tree-partition-width, which shows that it is unlikely that there exists an algorithm running in FPT time. 

We outline some relevant definitions and background in Section \ref{sec:prelim}.
The algorithm of Theorem \ref{thm:main} is highly non-trivial and requires several conceptual ideas that we have attempted to separate out by first explaining some key ideas and an easier deterministic algorithm that uses $O(\log^2 n)$ space in Section \ref{sec:warmup}.
We explain two ideas that need to be added to this, present the algorithm and give the space analysis in Section \ref{sec:proof}. We prove our results concerning tree-partition-width in Section \ref{sec:boundedtpw} and point to some directions for future work in Section \ref{sec:concl}. Some technical details are left to the appendix.
\section{Preliminaries}
\label{sec:prelim}
All logarithms in this paper have base 2. Let $T$ be a rooted tree and $v\in V(T)$. We write $T_v$ for the subtree rooted at $v$ and $T-v$ for the forest obtained by removing $v$ and all edges incident to $v$. 
We refer the reader to textbooks for basic notions in graph theory \cite{TextbookDiestel} and (parameterised) complexity \cite{TextbookAroraBarak,DowneyFellows}.

\subsection{Simple logspace computations on trees}
We repeatedly use the fact that simple computations can be done on a rooted tree using logarithmic space, such as counting the number of vertices in a subtree. We include a brief sketch below and refer to \cite{Lindelltreeisom} for further details.

We first explain how to traverse a tree in logspace.
Record the
index of the current vertex and create states \texttt{down}, \texttt{next} and \texttt{up}. We start on the root with state
\texttt{down}. When in state \texttt{down}, we go to the first child while remaining in the state \texttt{down}. If there is
no child, we change the state to \texttt{next}. When in state \texttt{next}, we go to the next sibling if
it exists and change state to \texttt{down}, or (if there is no next sibling) we change state to
\texttt{up}. When in state \texttt{up}, we simultaneously go to the parent and change the state to \texttt{next}. We stop when reaching the root with state \texttt{up}.
By keeping track of the number of vertices discovered, we can use the same technique to count the
number of vertices in a subtree. This can then be used to compute the child with maximum subtree
size and to enumerate children ordered by their subtree size.
If the input tree is not rooted, we may use the indices of the vertices to root the tree in a deterministic way.

\subsection{Graph width measures}

Let $G=(V,E)$ be a graph. A  tuple $(T,\{X_t\}_{t\in V(T)})$ is  a \emph{tree decomposition} for $G$ if $T$ is a tree, for $t \in V(T)$, $X_t \subseteq V$ is the \emph{bag} of $t$, for each edge $uv \in E$ there is a bag such that $\{u,v\} \subseteq X_t$, and for each $u \in V$, the bags containing $u$ form a nonempty subtree of $T$.
If $T$ is a path, this defines a \emph{path decomposition}.

The width of such a decomposition is $\max_{t \in V(T)} |X_t|-1$. The \emph{treewidth} (resp. \emph{pathwidth}) of $G$ is the minimum possible width of a tree decomposition (resp. path decomposition) of $G$.

A \emph{nice} path decomposition has empty bags on endpoints of the path and two consecutive bags differ by at most one vertex. Hence, either a vertex is \emph{introduced}, or a vertex is \emph{forgotten}.

Let $G$ be a graph,
let $T$ be a tree, and, for all $t\in V(T)$, let $X_t$ be a non-empty set so that $(X_t)_{t\in V(T)}$ partitions $V(G)$. The pair $(T,(X_t)_{t\in V(T)})$ is a \emph{tree-partition} of $G$ if, for every edge $vv'\in E(G)$, either $v$ and $v'$ are part of the same bag, or $v\in X_t$ and $v'\in X_{t'}$ for $tt'\in E(T)$. The width of the partition is $\max_{t\in V(T)}|X_t|$. The \textit{tree-partition-width} (also known as strong treewidth) of $G$ is
the minimum width of all tree-partitions of $G$. It was introduced by Seese \cite{Seese85} and can be characterised by forbidden topological minors \cite{DingO96}. Tree-partition-width is comparable to treewidth on graphs of  maximum degree $\Delta$ \cite{DingO95,Wood09}:
$\tw+1 \leq 2\tpw \leq O(\Delta\tw)$. However, it is incomparable to treedepth, pathwidth and treewidth for general graphs. 

The \emph{treedepth} of a graph is the minimum height of a forest $F$ with the property that every edge of $G$ connects a pair of nodes that have an ancestor-descendant relationship to each other in $F$.

\section{Warm up: first ideas and a simpler algorithm}
\label{sec:warmup}
\subsection{Storing colours via their position in the list}
\label{subsec:max_deg}
It is not too difficult to obtain a non-deterministic algorithm that uses $O(\log^2n)$ space.
\begin{proposition}
\label{prop:max_degree}
\textsc{List Colouring} can be solved non-deterministically
using $O(\log n \log \Delta)$ space on $n$-vertex trees of maximum degree $\Delta$.
\end{proposition}
The proposition follows from the following two lemmas, which are proved in Appendix \ref{app:a}.
\begin{lemma}
\label{lem:pathwidth}
\textsc{List Colouring} can be solved non-deterministically using $O(k\log \Delta +\log n)$ space for an $n$-vertex graph $G$ of maximum degree $\Delta$ if we can deterministically compute a path decomposition for $G$ of width $k$ in $O(\log n)$ space.  
\end{lemma}
A deterministic logspace algorithm for computing an optimal path decomposition exists for all graphs of bounded pathwidth \cite{pathwidthpaper}, but this does not apply directly to trees (since their pathwidth may grow with $n$). 
\begin{lemma}
\label{lem:pd}
    If $T$ is an $n$-vertex tree, we can deterministically construct a nice path decomposition of width $O(\log n)$ using
    $O(\log n)$ space.
\end{lemma}
We remark that $\Delta$ may be replaced by a bound on the list sizes in Proposition \ref{prop:max_degree} and Lemma \ref{lem:pathwidth}. The main observation in the proof of Lemma \ref{lem:pathwidth} is that for a vertex $v$, we only need to consider the first $d(v)+1$ colours from its list so that we can store the \emph{position} of the colour rather than the colour itself. 
Note that we cannot keep the path decomposition in memory, but rather recompute it whenever any information is needed. We keep in memory only the current position in the path decomposition and the list positions of the colours we assigned for vertices in the previous bag.

\subsection{Heavy children, recursive analysis and criticality}
Suppose that we are given an instance $(T,L)$ of \textsc{List Colouring}. We fix a root $v^*$ of $T$ in an arbitrary but deterministic fashion, for example the first vertex in the natural order on the vertices. 
Let $v\in V(T)$. We see $v$ as a descendant and ancestor of itself.
We write $T_{v}$ for the subtree with root $v$. 
\begin{definition}[Heavy]
A child $u$ of a vertex $v$ in a rooted tree $T$ is called \textit{heavy} if $|V(T_u)|\geq |V(T_{u'})|$ for all children $u'$ of $v$, with strict inequality whenever $u'<u$ in the natural order on $V$.
\end{definition}
Each vertex has at most one heavy child. We also record the following nice property.
\begin{observation}
\label{obs}
If $u$ is a child of $v$ which is non-heavy, then $|V(T_u)|\leq (|V(T_v)|-1)/2$.
\end{observation}
An obvious recursive approach is to loop over the possible colour $c\in L(r)$ for the root $r$ and then to recursively check for all children $v$ of $r$ whether a list colouring can be extended to the subtree $T_{v}$ (while not giving $v$ the colour $c$). We wish to prove a space upper bound of the form $S(n)=f(n)\log n$ on the number of bits of storage required for trees on $n$ vertices (for some non-decreasing function $f$). We compute
\begin{equation}
\label{eq:rec}
S(n/2)= \log (n/2)f(n/2)\leq \log(n/2)f(n)=\log nf(n)-\log 2f(n)\leq S(n)-f(n).
\end{equation}
This shows that while performing a recursive call on some subtree $T_v$ with $|V(T_v)|= n/2$,  we may keep an additional $f(n)$ bits in memory (on top of the space required in the recursive call). In particular, we can store the colour $c$ using $O(\log n)$ bits when $f(n) = \Theta( \log n)$, but can only keep a \emph{constant} number of bits for such recursions when proving Theorem \ref{thm:main}. 

We next explain how we can ensure that we only need to consider recursions done on non-heavy children.
Suppose $v$ has non-heavy children $v_1,\dots,v_k$ and heavy child $u$. We will write $G_v=T_v-T_u$.
Suppose the parent $v'$ of $v$ needs to be assigned colour $c'$. One of the following must be true.
\begin{itemize}
    \item There is no colouring of $G_v$ which avoids colour $c'$ for $v$. In this case, we can reject.
    \item There is a unique colour $c\neq c'$ that can be assigned to $v$ in a list colouring of $G_v$. We say $v$ is \emph{critical} and places the colour constraint on $u$ that it cannot receive colour $c$.
    \item There are two possible colours unequal to $c'$ that we can give $v$ in $G_v$. We then say $v$ is \textit{non-critical}: it can be coloured for each colour we might wish to assign its heavy child $u$. 
\end{itemize}
One example of non-criticality is if the list of available colours $|L(v)|\geq k+2$. In this case, $T_v$ is list colourable if and only if $T_{v_1},\dots,T_{v_k},T_u$ are all list colourable (using lists $L$).
\subsection{A deterministic algorithm using $O(\log^2n)$ space and polynomial time}
\label{subsec:warm-up}
We define a procedure \textsc{solve}$(v,p)$ which given a vertex $v$ and number $p\in \{0,1,\dots,n\}$, determines whether the subtree $T_v$ rooted on $v$ can be list coloured; for $p\geq 1$ there is an additional constraint that  $v$ cannot receive the $p$th colour in $L(v)$; for $p=0$ the vertex $v$ may receive any colour. 

Suppose we call \textsc{solve}($r,p$) for some $r\in V(T)$. Let $v_1,\dots,v_k$ denote the non-heavy children of $r$ and $u$ the heavy child.  
The algorithm works as follows.
\begin{enumerate}
    \item For $i=1,\dots,k$, we recursively verify that $T_{v_i}$ can be list coloured (\textsc{solve}$(v_i,0)$); we reject if any of these rejects.
    \item  If $|L(r)|\geq d(r)+1=k+2$, then we are `non-critical': we free up our memory (removing also $(r,p)$) and recursively verify whether $T_{u}$ can be list coloured by calling \textsc{solve}$(u,0)$.
    \item From now on, $|L(r)|\leq k+1$. We check that there is some $p_1\in \{1,\dots,|L(r)|\}$ for which we can assign $v$ the $p_1$th colour in its list, and extend to a list colouring of $T\setminus T_u$. This involves recursive calls \textsc{solve}$(v_i,p_{i,1})$ where $p_{i,1}$ places the appropriate constraint\footnote{Let $c_1$ be the $p_1$th colour in $L(r)$. If $c_1\not\in L(v_i)$, let $p_{i,1}=0$. Otherwise, let $p_{i,1}$ be the position of $c_1$ in $L(v_i)$.} on $v_i$.\\
If no such $p_1$ exists, we reject.
\item Next, we check whether $r$ is `non-critical', that is, whether there is some $p_2\neq p_1$ for which there is a list colouring of $T\setminus T_u$ in which $r$ receives the $p_2$th colour from its list.\\
If such $p_2$ exists, we free up our memory and recursively verify whether $T_{u}$ can be list coloured.\\
If $p_2$ does not exist, then we know that $r$ must get colour $p_1$. If $p_1=p$, we reject. \\
Otherwise, we free up our memory and run \textsc{solve}($u,p_1'$), where $p_1'$ is either $0$ or the position in the list of $u$ of the $p_1$th colour of $L(r)$.
\end{enumerate}
We give a brief sketch of the space complexity; more precise arguments including also pseudocode and the time analysis of this algorithm are given in Appendix \ref{app:b}. 

For the sake of analysis, suppose we keep a counter $\r$ that keeps track of the `recursion depth'. We increase this by one each time we do a call on a non-heavy child (decreasing it again once it finishes), but do not adjust it for calls on a heavy child.

Suppose $\r$ has been increased due to a sequence of recursive calls on $(v_1,0),(v_2,p_2),\dots,(v_\ell,p_\ell)$, with $v_1$ the root of the tree and $v_{i+1}$ a non-heavy child of some heavy descendant of $v_{i}$ for all $i\in[\ell-1]$. Then
$|V(T_{v_\ell})|\leq \frac12|V(T_{v_{\ell-1}})|\leq \dots \leq \frac1{2^{\ell-1}} |V(T)|$. In particular, $\r$ is always bounded by $\log n$.

Crucially, if a call  \textsc{solve}($u,p'$) is made on the heavy descendant $u$ of some $v_i$, the only information we need to store relating to the part of the tree `between $v_i$ and $u$' is $p'$. Therefore, if we distribute our work tape into $\lceil \log n\rceil$ parts where the $i$th part will be used whenever $\r$ takes the value $i$, then each part only needs to use $10 \log n$ bits, giving a total space complexity of $O(\log^2 n)$. 

\section{Proof of Theorem \ref{thm:main}}
\label{sec:proof}
In this section, we describe our $O(\log n)$ space algorithm. This also uses the ideas of using positions in a list (rather than the colours themselves), criticality and starting with the non-heavy children described in the previous section. However, we need to take the idea of first processing `less heavy' children even further. 

The main idea is to store the colour $c$ that we are trying for a vertex $v$ using the position $p_j$ of $c$ in some list $L_j(v)$, and to reduce the size of the list (and therefore the storage requirement of $p_j$) before we process `heavier' children of $v$. There are two key elements:
\begin{itemize}
    \item We can recompute $c$ from $p_j$ in $O(\log n)$ space. This is useful, since we can do a recursive call while only having $p_j$ (rather than $c$) as overhead, discover some information, and then recompute $c$ only at a point where we have a lot of memory available to us again.
    \item The space used for $p_j$ will depend on the size of the tree that we process. It is too expensive to consider all sizes separately, and therefore we will `bracket' the sizes. Subtrees whose size falls within the same bracket are processed in arbitrary order. For example, we put all trees of size $O(\sqrt{n})$ in a single bracket: these can be processed while using $O(\log n)$ bits of information about $c$ (which is trivially possible). 
    At some point we reach brackets for which the subtrees have size linear in $n$, say of size at least $\tfrac14 n$. Then, we may only keep a \emph{constant} number of bits of information about $c$. Intuitively, this is possible because at most four children of $r$ can have a subtree of size $\geq \tfrac14 n$, so the `remaining degree' of $v$ is small. In particular, if more than six colours for $v$ work for all smaller subtrees, then $v$ is `non-critical'.
\end{itemize}
We first explain our brackets in Section \ref{subsec:brack}. We then explain how we may point to a colour using less memory in Section \ref{subsec:pj} and how we keep track of vertices using less memory in Section \ref{subsec:pos}. We then sketch the proof of Theorem \ref{thm:main} by outlining the algorithm and its space analysis in Section \ref{subsec:alg}.
\subsection{Brackets}
\label{subsec:brack}
Recall that we fixed a root for $T$ in an arbitrary but deterministic fashion. 

Let $v\in V(T)$ and $u$ the heavy child of $v$. Let $G_v=T_v-T_u$ and $n_v=|V(T_v)|$. 
Each subtree $T'$ of $G_v-v$ is rooted in a non-heavy child of $v$ and will be associated to a bracket based on $|V(T')|$. By Observation \ref{obs}, $1\leq |V(T')| \leq (n_v-1)/2$.
Let $\M_v=\lceil \log\log(n_v/2)\rceil$. 
The brackets are given by the sets of integers in the intervals
\[
[1,n_v/2^{2^{\M_v-1}}), [n_v/2^{2^{\M_v-1}},n_v/2^{2^{\M_v-2}}), \dots,[n_v/256,n_v/16), [n_v/16,n_v/4), [n_v/4,n_v/2).
\]

There are $\M_v$ brackets: $[n_v/2^{2^{j}},n_v/2^{2^{j-1}})$ is the $j$th bracket for $j\in \{1,\dots,
\M_v-1\}$ and $[1,n_v/2^{2^{\M_v-1}})$ is the $\M_v$th bracket. Note that
$n_v/2^{2^{\M_v-1}}=O(\sqrt{n_v})$. This implies that while doing a recursive call on a tree in the $\M_v$th bracket, we are happy to keep an additional $O(\log n_v)$ bits in memory. 

We aim to show that for some universal constant $C$, when doing a recursive call on a tree in the $j$th bracket, we can save all counters relevant to the current call using at most $C2^{j}$ bits (which depending on the value of $j$, could be $\Omega(\log n)$). In our analysis, we save the counters in a new read-only part of the work tape. The recursive call cannot alter this (and will have to work with less space on the work space). We can then use our saved state to continue with our calculations once the recursive call finishes.

We short-cut  $\M=\M_v$ for legibility; the dependence of $\M$ on $v$ is only needed to ensure that we do not start storing counters for lots of empty brackets when $n_v$ is much smaller than $n$, and can be mostly ignored.

\subsection{The information $p_{j}$ stored about a colour $c$}
\label{subsec:pj}
Let $v\in V(T)$ with heavy child $u$ and $c\in L(v)$. Recall that $G_v=T_v-T_u$. 
We will loop over $j=\M,\dots,0$ and consider subtrees of $G_v$ whose size falls in the $j$th bracket in an arbitrary, but deterministic order (e.g. using the natural order on their roots). When $j$ decreases, we will perform recursions on larger subtrees of $G_v-v$ and can therefore keep less information about $c$. 
We first define `implicit' lists. \begin{itemize}
    \item Set $L_{\M}(v)=L(v)$. 
    \item For $j\in [0,\M-1]$, let $L_j(v)$ be the set of colours $\alpha\in L(v)$ such that all subtrees $T'$ of $G_v-v$ with $|V(T')|<n/2^{2^j}$ can be coloured without giving the colour $\alpha$ to the root of $T'$ (so that $v$ may receive colour $\alpha$ according to those subtrees).
\end{itemize}
Note that $L_j(v)=L(v)$ if there are no subtrees $T'$  with $|V(T')|<n/2^{2^j}$. Since the subtrees associated to brackets $1,\dots, j$ have size at least $n/2^{2^j}$, there can be at most $2^{2^j}$ of them. 
If $|L_j(v)|\geq 2^{2^j}+3$ and all subtrees $T'$ of $G_v-v$ can be coloured, then $v$ is `non-critical': after the parent and heavy child of $v$ have been coloured, the colouring can always be extended to $v$ and the rest of $G_v$. 

Suppose we are testing if we can give colour $c$ to $v$. If $c\not\in L_j(v)$, then we may reject: $c$ is not a good colour for one of the subtrees. We define $p_{j}=p_j(c,v)$ as the integer $x\in \{1,\dots,|L_j(v)|\}$ such that the $x$th element in $L_j(v)$ equals $c$. In particular, $p_{\M}$ is the position of the colour $c$ in the list $L_{\M}(v)=L(v)$ and $p_0$ gives the position of $c$ in the list of colours for which all subtrees of $G_v$ allow $v$ to receive $c$. For $j<\M$, we will reserve at least $\log(2^{2^j}+3)$ bits for $p_{j}$. This is possible, because we can maintain that $|L_j(v)| \leq 2^{2^j}+2$ by going into a `non-critical subroutine' if we discover the list is larger.

\subsection{Position of the current vertex}
\label{subsec:pos}
Next, we describe how to obtain efficient descriptions of the vertices in the tree. When performing recursions, we find it convenient to store information using which we can retrieve the `current vertex' of the parent call. Therefore, we require small descriptions for such vertices if the call did not make much `progress'.

For any $v\in V(T)$, define a sequence $h(v,1),h(v,2),\dots$ of heavy descendants as follows. Let $h(v,1)=v$.
Having defined $h(v,i)$ for some $i\geq 1$, if this is not a leaf, we let $h(v,i+1)$ be the heavy child of $h(v,i)$. Note that given the vertex $h(v,i)$, we can find the vertex $h(v,i+1)$ (or conclude it does not exist) in $O(\log n)$ space. 
We give a deterministic way of determining a bit string $\poss(v,i)$ that represents $h(v,i)$, where the size of the bit string will depend on the `progress' made at the vertex $h(v,i)$ that it represents. 
This `progress' is measured by the size $t_i$ of the largest subtree $T'$ of a non-heavy child of $h(v,i)$, that is, $T'$ is the largest component of $T-h(v,i)$ which does not contain $h(v,j)$ for $j\neq i$.
We define $\poss(v,i)$ as follows. 
\begin{itemize}
    \item Let $j$ be given such that $t_i\in [n_v/2^{2^{j}},n_v/2^{2^{j-1}})$. Start $\poss(v,i)$ with $j$ zeros, followed by a $1$.
    \item There are at most $2^{2^{j}}$ values of $i$ for which $t_i\geq 2^{2^{j}}$. We add a bit string of length $2^j$ to $\poss(v,i)$, e.g. the value $x$ for which $h(v,i)$ is the $x$th among $h(v,1),\dots,h(v,a)$ with $t_i\in [n_v/2^{2^{j}},n_v/2^{2^{j-1}})$.
\end{itemize}
Note that, given $v$, we can compute $\poss(v,i)$ from $h(v,i)$ and $h(v,i)$ from $\poss(v,i)$ using $O(\log n)$ space. 
If we do a recursive call, it will be on a non-heavy child $u$ of some $h(v,i)$. By definition, $|V(T_u)|\leq t_i$ and $\poss(v,i)$ depends on $t_i$ in a way that we are able to keep it in memory while doing the recursive call.

We use the encoding $(\poss(v^*,i),j,\ell)$ for the $\ell$th child $u$ of $h(v^*,i)$ whose subtree has a size that falls in the $j$th bracket. We can also attach another such encoding, e.g.
\[
((\poss(v^*,i),j,\ell),(\poss(u,i'),j',\ell')),
\]
to keep track of $u$ and the $\ell'$th child $v$ of $h(u,i')$ whose subtree has a size that falls in the $j'$th bracket. We can retrieve $v$ from the encoding above in $O(\log n)$ space and therefore can retrieve it whenever we have such space available to us.

\subsection{Description of the algorithm}
\label{subsec:alg}
During the algorithm, the work tape will always start with the following.
\begin{itemize}
    \item $\r$: the recursion depth $r$ written in unary. At the start, $r=0$.
    \item $\pos=\pos_0|\cdots |\pos_r$: encodes vertices as described Section \ref{subsec:pos}. At the start, this is empty and points at the root $v^*$ of the tree on the input tape. 
    \item $\colres=\colres_0|\dots|\colres_r$: encodes colour restriction information for the vertices encoded by $\pos$. At the start, this is empty and no restrictions are given. We maintain throughout the algorithm that $\colres_i$ gives us colour restrictions for the vertex $v$ pointed at by $\pos_i$. The restriction can either be `no restrictions' or `avoid $c'$'; in the latter case $\colres_i$ contains a tuple $(j,p_j')$ with $p_j'$ the position of $c'$ in $L_j(v')$, where $v'$ is the parent of $v$.
    \item $\auxi=\auxi_0|\dots|\auxi_r$: further auxiliary information for parent calls that may not be overwritten.
\end{itemize}
We define a procedure $\texttt{process}$. While the value of $\r$ equals $r$, the bits allocated to $\pos_i,\colres_i,\auxi_i$ for $i<r$ will be seen as part of the read-only input-tape.   In particular, the algorithm will not make any changes to $\pos_i,\colres_i$ or $\auxi_i$ for any $i<r$, but may change the values for $i=r$.

We will only increase $\r$ when we do a recursive call. If during the run of the algorithm $\r=r$ and a recursive call is placed, then we increase $\r$ to $r+1$ and return to the start of our instructions. However, since $\r$ has increased it will now see a `different input tape'.  When the call finished, we will decrease $\r$ back to $r$ and wipe everything from the work space except for $\r,\pos_i,\colres_i,\auxi_i$ for $i\leq r$, and the answer of the recursive call ($0$ or $1$). We then use $\auxi_i$ to reset our work space and continue our calculations.

We will ensure that $\r$ is always upper bounded by $\log n$. Indeed, the vertex $v_r$ encoded by $\pos_r$ will always be a non-heavy child of a descendant of the vertex encoded by $\pos_{r-1}$. 

While $\r=r$, the algorithm is currently doing calculations to determine whether the vertex $v$ pointed at by $\pos_{r-1}$ ($v^*$ for $r=0$) has the property that $T_{v}$ can be list coloured, while respecting the colour restrictions encoded by $\colres_{r-1}$ (none if $r=0$). 
Recall that rather than writing down $v$ explicitly, we use a special encoding from which we can recompute $v$ whenever we have $C\log n$ space available on the work tape (for some universal constant $C$). 
Similarly, $\colres_{r-1}$ may give a position $p_j'$ in $L_j(v')$ for some $j<\M$ (for $v'$ the parent of $v$), and we may need to use our current work tape to recompute the corresponding position $p_\M'$ of the colour in $L(v')$, so that we can access the colour from the input tape.
This part will make the whole analysis significantly more technical. 

We define an algorithm which we call $\texttt{process}$ as follows. A detailed outline is given in Appendix \ref{app:c}, whereas an informal description of the steps is given below.
\begin{enumerate}
    \item[0.] Let $v$ be the vertex pointed at by $\pos_\r$. We maintain that at most one colour $c'$ has been encoded that $v$ must avoid. Handle the case in which $v$ is a leaf. If not, it has some heavy child $h$. We go to 1, which will eventually lead us to one of the following (recall that $G_v=T_v-T_h$):
    \begin{enumerate}
    \item[(rej)] There is no list colouring of $G_v$ avoiding $c'$ for $v$. We return \texttt{false}.
    \item[(nc)] The vertex $v$ can get two colours (unequal to $c'$) in list colourings of $G_v$. In this case, we say $v$ is non-critical. We update $\pos_\r$ to $h$, set $\colres_\r$ to `none' and repeat from 0.
    \item[(cr)] There is a unique colour $c\neq c'$ that works. Then $\{c\}\subseteq L_0(v)\subseteq \{c,c'\}$ and so can represent $p_0=p_0(c,v)$ with a single bit. We update $\pos_\r$ to $h$, update $\colres_\r$ to $(0,p_0)$ and repeat from 0.
\end{enumerate}
    \item We check that all subtrees can be coloured if we do not have any colour restrictions (necessary for the non-critical subroutine). This involves recursive calls on $\process$ where we have no colour constraint on the root of the subtree.
    \item We verify that $L_0(v)$ is non-empty. We iteratively try to compute $p_\M$ from $p_{0}=1$ via $p_{1},\dots,p_{\M-1}$. Starting from $p_0=1$ and $j=0$, we compute $p_{j+1}$ from $p_{j}$ as follows.
\begin{enumerate}[(i)]
    \item Initialise $\texttt{curr}_j=1$. This represents a position in $L_j(v)$ (giving the number of `successes').\\
Initialise $\texttt{prev}_j=1$. This represents a position in $L_{j+1}(v)$ (giving the number of `tries').  
    \item We check whether the $\texttt{prev}_j$th colour of $L_{j+1}(v)$ works for the trees in the $j$th bracket. This involves a recursive call on $\process$ for each tree $T'$ in the $j$th bracket, putting $(j+1,\texttt{prev}_j)$ as the colour constraint for the root of $T'$. (The colour restriction gives a position in $L_{j+1}(v)$; we do not store the corresponding colour and rather will recompute it in the recursive call!)
    \item If one of those runs fails, we increase $\texttt{prev}_j$; if this is now $>2^{2^{j+1}}+3$, then this implies a lower bound on $|L_{j+1}(v)|$ which allows us to move to (nc). \\
    If $\texttt{curr}_j<p_{j}$
    we increase both $\texttt{curr}_j$ and $\texttt{prev}_j$. 
    
    Once $\texttt{curr}_j=p_{j}$, we have succesfully computed $p_{j+1}=\texttt{prev}_j$ and continue to compute $p_{j+2}$ if $j+1<\M$. Otherwise we repeat from (ii).
\end{enumerate}
Once we have $p_\M$, we can access the corresponding colour from the input tape by looking at the $p_\M$th colour of $L(v)$. If $p_\M>|L(v)|$, we go to (rej).
\item We verified $|L_0(v)|\geq 1$. We establish whether $|L_0(v)|\geq 3$ in a similar manner. If so, we can go to (nc); else we need to start considering the colour constraints of the parent $v'$ of $v$. Note that we can use $\auxi_\r$ to store auxiliaries such as $\alpha=|L_0(v)|\in \{1,2\}$.
    
It remains to explain how we check whether the first or second colour from $L_0(v)$ satisfies the colour constraint from $v'$. Suppose a colour $c'$ has been encoded via the position $p_{j'}'$ of $c'$ in $L_j(v')$ for some $j'\in [0,\M']$ (where $\M'=\M_{v'}$). 

We can recompute $p_{\M'}'$ from $p_j'$ in a similar manner to the above. However, once we store $p_{\M'}'$, we may no longer be able to compute $p_\M$ from $p_0=1$ or $p_0=2$, since $p_{\M'}'$ may take too much space\footnote{This is one of the issues that made this write-up more technical and involved than one might expect necessary at first sight; if we do a recursion on a child in the $j$th bracket of $v$ or $v'$, then we are only allowed to keep $O(2^{j})$ bits on top of the space used by the recursion. This means we cannot simply keep $p_{\M'}'$ in memory if $j$ is much smaller than $\M'$.}. Therefore, we first recompute $p_{j'}$ from $p_0$ and then simultaneously recompute $p_{j'+1}$ from $p_{j'}$ and $p'_{j'+1}$ from $p'_{j'}$ until we computed $p_\M$ and $p_{\M'}'$. We then check whether the $p_\M$th colour of $L(v)$ equals $c'$, the $p_{\M'}'$th colour of $L(v')$. (It may be that $\M \neq  \M'$, meaning that we may finish one before the other.)

The computation of $p_{x+1}'$ from $p_{x}'$ for the parent $v'$ of $v$ is a bit more complicated if $v$ is a non-heavy child of $v'$. In this case, $v$ is in the $(j'-1)$th bracket of $v'$. The algorithm calls again on itself for subtrees in the $x$th bracket of $v'$, but now we see a resulting call to \texttt{process} as a \emph{same-depth call} rather than a `recursive call'. The computations work the same way, but we do not adjust $\r$ and will \emph{add} the current state from before the call in $\texttt{aux}_\r$.
\end{enumerate}
By an exhaustive case analysis, the algorithm computes the right answer. It remains to argue that it terminates and runs in the correct space complexity.

We first further explain the same-depth calls on \texttt{process}. When we make such call, we do not adjust $\r$.  When the call is made, $\pos_\r$ will encode a vertex $v_1$ which is a non-heavy child of $v'$. After the call, $\pos_\r$ will point to some other non-heavy child $v_2$ of $v'$. Importantly, the bracket of $v_2$ will always be higher than the bracket of $v_1$, so that at most $\M$ such same-depth calls are made before changing our recursion level. Each same-depth call may stack on a number of auxiliary counters which we keep track of using $\auxi_{\r}$; since for each $j$, there is at most one vertex from the $j$th bracket which may append something to this, it suffices to ensure a vertex from bracket $j$ adds at most $C2^j$ bits. Indeed, this ensures that the size of $\auxi_r$ takes at most $C\sum_{j=1}^{j'}2^j\leq C2^{j'+1}=O(2^{j'})$ bits once $\pos_r$ encodes a vertex from bracket $j'$.

We now argue that it terminates. 
Each time we do a same depth call, the value of the bracket $j$ will increase by at least one. We therefore may only do a finite number of same-depth calls in a row. 
When $\r=0$, each time we reach (nc) or (cr), we move one step on the heavy path from the root to a leaf, so this will eventually terminate. A similar observation holds for $\r>0$ once we fix $\pos_1|\dots|\pos_{\r-1}$: this points to some vertex $v$ and $\pos_r$ will initially point to a non-heavy descendent $u$ or $v$, and then `travel down the heavy path' from $u$ to a leaf. 

We next consider the space used by the algorithm. 
Let $S(n)$ be the largest amount of bits used for storing
$\r,\pos,\colres$ and $\auxi$ during a run of \textsc{process} on an $n$-vertex tree. We can distribute our work space into two parts: $C_1\log n$ space for temporary counters and for doing calculations such as computing the heavy child of a vertex, and $S(n)$ bits for storing $\r,\pos,\colres$ and $\auxi$. 
It suffices to prove that $S(n)\leq C_2\log n$; this is the way we decided to formalise keeping track of the `overhead' caused by recursive calls.

Note that $\r$ is bounded by $\log n$ if the input tree has $n$ vertices: each time it increases, we moved to a non-heavy child whose subtree consists of at most $n/2$ vertices. 

Suppose we call \process with $\pos$ pointing at some vertex $v$ whose subtree has size $n_v$. We will show inductively that the number of bits used by $\pos,\colres$ and $\auxi$ is in $O(\log n_v)$ throughout this call (note that $n_v$ may be much smaller than $n$, the number of vertices of the tree on the input tape). Whenever we do a recursive call, this will be done on a tree whose size is upper bounded by $n_v/2^{2^{j-1}}$ for some $1\leq j\leq \M=\lceil \log\log(n_v/2)\rceil$. Since by induction the recursive call requires only
\[
S(n/2^{2^{j-1}})=C_2\log(n/2^{2^{j-1}})\leq S(n)-C_22^{j-1}
\]
additional bits, we will allow ourselves to add at most $C_22^{j-1}$ bits to $\r,\pos,\colres$ and $\auxi$ before we do a recursive call that divides the number of vertices by at least $2^{2^{j-1}}$. The constant $C_2$ will be relatively small (for example, 1000 works). 

Fix a value of $j$.
From the definitions in Section \ref{subsec:pj} and \ref{subsec:pos}, at the point that we do a recursion on a subtree whose size is upper bounded by $n_v/2^{2^{j-1}}$, $\pos_\r$ and $\colres_\r$ store at most three integers (they are of the form $(\poss(v',i),j,\ell)$ and $(x,p_x)$ respectively) that are bounded by $2^{2^j}+3$. Therefore, these require at most $C'2^{j-1}$ bits for some constant $C'$ (e.g. $C'=50$ works). 

The worst case comes from $\auxi_\r$ which may get stacked up on by the `same-depth calls'. There is at most one such same-depth call per $x\in \{1,\dots,j\}$. For such $x$, we add on a bounded number of counters (e.g. $\texttt{curr}_x$ and $\texttt{prev}_x$) which can take at most $2^{2^{x+1}}+3\leq 2\cdot 2^{2^{x+1}}$ values. Since $C\sum_{x=1}^j 2^x\leq 4C2^{j-1}$, $\auxi_\r$ also never contains more than $O(2^{j-1})$ bits while doing a call on a tree whose size falls in bracket $j$.

\section{Graphs of bounded tree-partition-width}
\label{sec:boundedtpw}

\subsection{Proof of Corollary \ref{cor:listcol_tpw}}

Here we sketch the proof of Corollary \ref{cor:listcol_tpw}.
Let $(T,(X_t)_{t\in V(T)})$ be a tree-partition of width $k$ for an $N$-vertex graph $G$, where it will be convenient to write $n=|V(T)|\leq N$. We prove that there is an algorithm running in $O(k\log k\log n)$ space, for each $k$ by induction on $n$. 

We root $T$ in the vertex of lowest index.
We define the recursion depth, heavy children, the brackets, $M$ and `position' for the vertices in $T$ exactly as we did in Section \ref{sec:proof}. 

Suppose $\pos$ points at some vertex $t\in V(T)$. We aim to use $O(k\log k\log(2^{2^j}))$ bits for the information $p_j$ stored about the colouring $c$ of the vertices in $X_t$ while processing subtrees in bracket $j$. This is done as follows. For $t\in V(T)$, let $G_{t,j}$ be the graph induced on the vertices that are either in $X_t$ or in $X_{s}$ for $s$ a non-heavy child of $t$ in bracket $j$ or above. 
For $v\in X_t$, we define $L_j(v)$ to be the list of colours $\alpha \in L(v)$ such that there is a list colouring of $G_{t,j}$ that assigns the colour $\alpha$ to $v$. There are at most $2^{2^j}$ subtrees of $T$ associated to brackets $1,\dots,j$, and so these include at most $k2^{2^j}$ neighbours of $v$. Therefore, if $|L_j(v)|\geq k(2^{2^j}+3)$, then $G$ can be list coloured if and only if $G-v$ can be list coloured and we no longer need to keep track of the colour of $v$. We then establish $v$ is non-critical. 
We use $k$ bits to write for each vertex of $X_t$ whether or not it has been established to be critical, and for those vertices that are critical, we use $\log(k(2^{2^j}+3))$ bits per vertex to index a colour from $L_j(v)$.

We use $k + \log((3k)^k)=O(k\log k)$ bits of information in $\auxi$ to keep track of the following.
\begin{itemize}
    \item For each $v\in X_t$, whether or not $v$ has been established to be critical. After processing $t$, the vertex $v\in X_t$ will be critical if there are at most $3k$ colours in $L(v)$ for which there exists an extension to $G_t$. Let $C_t\subseteq X_t$ denote the critical vertices.
    \item For each $p_0\in \prod_{v\in C_t}L_0(v)$ (of which there are at most $(3k)^k$), a single bit which indicates whether or not the parent of $t$ would allow the corresponding colouring.
\end{itemize}
While computing the information above, we still need the auxiliary information from the parent of $t$, but we can discard this by the point we start processing the heavy child of $t$.

We make two more small remarks:
\begin{itemize}
    \item We need to redefine what we mean by `increasing' $p_j$ for some $j$, since we now work with a tuple of list positions. We fix an arbitrary but deterministic way to do this, for example in lexicographical order using the natural orders on the vertices and colours.
    \item When we check whether the colouring $c$ of $X_t$ corresponding to $p_0\in \prod_{v\in C_t}L_0(v)$ is allowed by the parent $t'$ of $t$, we run over $p_0'\in \prod_{v'\in C_{t'}}L_0(v')$, and as before we need to compute $p_i,p_i'$ from $p_{i-1},p_{i-1}'$ iteratively until we obtain the colourings corresponding to $p_\M$ and $p_\M'$. If those colourings are compatible, then we know that there is a list colouring of the graph `above $t$' for which $X_t$ is coloured `according to $p_0$', and so we record in {\auxi} that $p_0$ is allowed. 
\end{itemize}
The calculations in the space analysis work in the exact same way: we simply multiply everything by $C k\log k$ (for a universal constant $C$).

\subsection{W[1]-hardness}
We give an easy reduction for the following result.
\begin{theorem}
    \textsc{List Colouring} parameterised by the width of a given tree-partition is \textup{W[1]}-hard.
\end{theorem}

\begin{proof}
    We reduce from \textsc{Multicoloured Clique}.
    
    Consider a \textsc{Multicoloured Clique} instance $G=(V,E),V_1,\ldots,V_k$ with $k\geq 2$ colours.
    We denote by $\overline{G}=(V,\overline{E})$ the complement of $G$.
    
    We now describe the construction of our instance graph $H$.
    We first add vertices $v_1,\ldots,v_k$, with lists $L(v_i)=V_i$ for all $i\in [k]$.
    Then for each edge $e=uv \in \overline{E}$ we add a vertex $x_{uv}$ with list $L(x_{uv})=\{u,v\}$. Furthermore, for $\alpha \in \{u,v\}$, let $i$ be such that $\alpha \in V_i$. We add the edge $x_{uv} v_i$. 
    
    The resulting graph has tree-partition-width at most $k$: we put $v_1,\ldots,v_k$ in the same bag which is placed at the centre of a star, and create a separate leaf bag containing $x_{uv}$ for each $uv\in \overline{E}$.
    
    \begin{claim}
        If there is a proper list colouring of $H$, then there is a multicoloured clique in $G$.
    \end{claim}
    
    \begin{proof}
    Suppose that $H$ admits a list colouring. Let $a_i\in V_i=L(v_i)$ be the colour assigned to $v_i$ for all $i\in [k]$. We will prove $a_1,\dots,a_k$ forms a multicoloured clique in $G$.
        
    Consider distinct $i,j \in [k]$ and suppose $a_ia_j$ is not an edge of $G$, that is, $a_ia_j\in \overline{E}$. Then there exists a  vertex $x_{a_ia_j}$ adjacent to both $v_i$ and $v_j$, but there is no way to properly colour it, a contradiction. So we must have $a_ia_j \in E$ as desired.
    \end{proof}
    
    \begin{claim}
        If there is a multicoloured clique in $G$, then there is a proper list colouring of $H$.
    \end{claim}
    
    \begin{proof}
        We denote by $a_1,\ldots,a_k$ the vertices of the multicoloured clique, where $a_i \in V_i$ for all $i$.
        We assign the colour $a_i$ to vertex $v_i$. Consider now $x_{uv}$ for some $uv\in \overline{E}$. Let $i$ and $j$ be given so that $x_{uv}$ is adjacent to $v_i$ and $v_j$. Then $\{u,v\}\neq \{a_i,a_j\}$ since $a_ia_j\in E$ and $uv\in \overline{E}$. Therefore, we may assign either $u$ or $v$ (or both) as colour to $x_{uv}$.
    \end{proof}
    Since \textsc{Multicoloured Clique} is W[1]-hard, this proves that \textsc{List Colouring} parameterised by tree-partition-width is W[1]-hard. 
\end{proof}
We remark that the above proof also shows that \textsc{List Colouring} parameterised by vertex cover is W[1]-hard.

\section{Conclusion}
\label{sec:concl}
In this paper, we combined combinatorial insights and algorithmic tricks to give a space-efficient colouring algorithm. 

By combining Logspace Bodlaender's theorem \cite{ElberfeldJT10}, Lemma \ref{lem:pd} and Lemma \ref{lem:pathwidth},  \textsc{List Colouring} can be solved non-deterministically on graphs of pathwidth $k$ in $O(k\log n)$ space and on graphs of treewidth $k$ in $O(k\log^2 n)$ space.\footnote{First compute a tree decomposition $(T,(B_t)_{t\in V(T)})$ of width $k$ for $G$ in $O(\log n)$ space \cite{ElberfeldJT10}, and then compute a (not necessarily optimal) path decomposition of width $O(\log |V(T)|)$ in $O(\log n)$ space for $T$, and turn this into a path decomposition for $G$ of width $O(k\log n)$ by replacing $t\in V(T)$ with the vertices in its bag $B_t$. Then use Lemma \ref{lem:pathwidth}.}
However, we already do not know the answer to the following question.
\begin{question}
Can a non-deterministic Turing machine solve \textsc{List Colouring} for $n$-vertex graphs of treewidth 2 using $o(\log^2 n)$ space?
\end{question}
Another natural way to extend trees is by considering graphs of bounded treedepth. Such graphs then also have bounded pathwidth (but the reverse may be false). It has been observed for several problems such as \textsc{3-Colouring} and \textsc{Dominating Set} that `dynamic programming approaches' (common for pathwidth or treewidth) require space exponential in the width parameter, whereas there is a `branching approach' with space polynomial in treedepth \cite{ChenRRV18}. A simple branching approach also allows \textsc{List Colouring} to be solved in $O(k\log n)$ space on $n$-vertex graphs of treedepth $k$. We wonder if the approach in our paper can be adapted to improve this further.
\begin{question}
Can \textsc{List Colouring} be solved in $f(k)g(n)+O(\log n)$ space on graphs of treedepth $k$, with $g(n)=o(\log n)$ and $f$ a computable function?
\end{question}
Another interesting direction is what the correct complexity class is for \textsc{List Colouring} parameterised by tree partition width. We do not expect this to be in the W-hierarchy because the required witness size seems to be too large. Moreover, the conjecture \cite[Conjecture 2.1]{PW18} mentioned in the introduction together with Corollary \ref{cor:listcol_tpw} would imply that the problem is not XNLP-hard.

Finally, it would be interesting to study other computational problems than \textsc{List Colouring}. We remark that our results are highly unlikely to generalise to arbitrary Constraint Satisfaction Problems. 
Recall that there is conjectured to be a $k^*\in \N$ for which \textsc{List Colouring} requires $\omega(\log n)$ space for $n$-vertex graphs of treewidth $k^*$.
Since \textsc{List Colouring} on $n$-vertex graphs of treewidth at most $k^*$ can be reduced in logspace to a CSP on at most $n$ variables, each having a lists of size at most $n^{k^*}$, and binary constraints on the variables, such CSP problems would then also require $\omega(\log n)$ space since $k^*$ is a constant.

\paragraph{Acknowledgements} We thank Marcin Pilipczuk and Michał Pilipczuk for discussions.
\bibliographystyle{plain}
\bibliography{refs}

\appendix
\section{Missing proofs from Section \ref{subsec:max_deg}}
\label{app:a}
\subsection{Proof of Lemma \ref{lem:pathwidth}}
Let an $n$-vertex graph $G$ of pathwidth $k$ be given. 

There is a (deterministic) logspace transducer for turning any path decomposition into a nice path decomposition of the same width.
We briefly describe such a transducer. First, introduce the vertices of the first bag. Then, for each new bag, forget the vertices that were in the previous bag but are not in this bag, and then introduce the vertices that are in this bag but not the previous bag. Conclude by forgetting the vertices of the last bag. The variables required are a pointer to the current bag, and two pointers to parse consecutive bags concurrently.

Consequently, we may assume that we can compute a nice path decomposition $(X_1,\dots,X_\ell)$ of width $k$ for $G$ in $O(\log n)$ space. We cannot store this path decomposition, but we can reserve an additional $O(\log n)$ bits on our work space to recompute the relevant part at the points we need it.

We compute and store $\ell$, the number of bags in the path decomposition. We initialise $i=1$; this indicates the bag we are currently considering.

We will also keep track of integers $p_{j},p_j'\in [0,\Delta+1]$ for $j\in [k+1]$. We use the value $0$ for $p_j$ if $X_{i-1}$ has no vertex at position $j$; otherwise, the $j$th vertex of $X_{i-1}$ has been `assigned' the $p_j$th colour in its list. The $p_j'$ are interpreted analogously for the current bag $X_i$ instead of the previous bag $X_{i-1}$. 

We initialise $i=1$ and $p_j=p_j'=0$ for $j\in [k+1]$.
For $i=1,\dots, \ell$, we see whether $X_i$ is a forget bag or an introduce bag.
\begin{itemize}
    \item If $X_i$ is an introduce bag, then we guess an integer $p\in [\Delta+1]$ for the new vertex in $X_i$.
    \item If $X_i$ is a forget bag, we check that the vertex that has been forgotten from $X_{i-1}$ has a different colour from all vertices in $X_{i-1}$ that it is adjacent to.
\end{itemize} 
The $p_j,c_j$ take $O(k\log \Delta)$ bits; an additional $O(\log n)$ bits of space are used for computations and auxiliaries. 

\subsection{Proof of Lemma \ref{lem:pd}}
We will parse the tree using logarithmic space and keep a stack with the index of the guessed colours for the vertices on the path from the root to the current vertex for which we have not yet explored all of their children.
These vertices and the current vertex constitute the current bag of a path decomposition. A vertex
that has been included in the partial decomposition can only be removed once all of its incident
edges have been covered. Hence, these vertices are those that cannot be removed yet.
Since we parse the whole graph, the consecutive values of the stack correspond to a path
decomposition. 
Note that since bag vertices will be removed one by one when parsing the tree, we actually have a nice path decomposition.

We now argue that, by always choosing to explore the largest subtree last (which can be done using deterministic logspace subroutines), we can ensure that the bag has only $O(\log|V(T)|)$ vertices.
We prove this by induction. Consider the root $r$ of tree $T$.
The largest subtree has size at most $|V(T)|-1$, but the current root is removed from the stack
when exploring the child with largest subtree.
A child which does not have the largest subtree has size at most $|V(T)|/2 -1$.
It follows that the maximum number of vertices in the bag for an $n$-vertex tree is $s(n) \leq
\max(s(n-1),1+s(n/2-1))$. We conclude that there are at most $\log n$ vertices in the bag.

\section{Additional details for Section \ref{subsec:warm-up}}
\label{app:b}
We denote the list of children of $v$ by $D(v)$, and denote by \texttt{largest} a function that, given a vertex of the tree, returns a child with largest subtree. 
A precise description is given in Algorithm \ref{algo:l2-algo}.

\begin{figure}[ht]
\centering
\begin{minipage}{0.6\linewidth}
\begin{algorithm}[H]
\DontPrintSemicolon
\SetKwFunction{Solve}{Solve}
\SetKwProg{Fn}{}{:}{}
\SetKwIF{If}{ElseIf}{Else}{if}{:}{else if}{else :}{}
\SetKwFor{For}{for}{:}{}
\Fn{\Solve{v,p}}{
\eIf{$|L(v)-\{c\}|>|D(v)|$}
{\eIf{$\forall u \in D(v)$, \Solve{$u,-1$}}{\KwRet \textbf{true}}{\KwRet \textbf{false}}}
{$c' \leftarrow -1$\;
$u' \leftarrow \texttt{largest}(v)$\;
\For{$u \in D(v)-\{u'\}$}
{\If{$\neg$ \Solve{$u,-1$}}{\KwRet \textbf{false}}}
\For{$a \in L(v)-\{c\}$}
{\For(\textit{(ordered by subtree size)}){$u \in D(v)-\{u'\}$}{
\If{$\neg$ \Solve{$u,a$}}{try next $a$}}
\eIf{$c' = -1$}{$c' \leftarrow a$}
{\textbf{free}($v,c,a,u,c'$)\;
\KwRet \Solve{$u',-1$}}
}
\eIf{$c'=-1$}{\KwRet \textbf{false}}
{\textbf{free}($v,c$)\;
\KwRet \Solve{$u',c'$}}
}
}

\caption{Simple algorithm using $O(\log(n)^2)$ space}
\label{algo:l2-algo}
\end{algorithm}
\end{minipage}
\end{figure}

\begin{theorem}
\label{thm:first}
Algorithm \ref{algo:l2-algo} solves \textsc{List Colouring} on $n$-vertex trees in polynomial time and $O(\log^2n)$ space.
\end{theorem}

\begin{proof}
Correctness is easy to verify. 
\begin{claim}
    The algorithm uses $O(\log^2 n)$ bits of space.
\end{claim}

\begin{proof}
    Denote by $S(n)$ the maximum amount of memory used for a tree of order $n$.
    The largest subtree, and the next subtree in the size ordering can be computed in $O(\log n)$ space and linear time. The local variables require $O(\log n)$ memory, and we free their space before recursing to the largest subtree. We conclude that we have $S(n) \leq \max( O(\log n) + S(n/2), S(n-1))$.
    The master theorem \cite{TextbookMasterTheorem} allows us to conclude that $S(n) = O(\log^2 n)$.
\end{proof}

\begin{claim}
    The algorithm runs in polynomial time.
\end{claim}

\begin{proof}
    Denote by $T(n)$ the maximum amount of time used for a tree of order $n$.
    Let $r$ be the root of the tree and $v_1,\ldots,v_k,v_{k+1}=u$ be its  children ordered by non-decreasing size of their subtree, with ties broken using the order given by the input. We denote by $x_i$ the size of the subtree of root $v_i$. Note that unlike the other subtrees, there is only one recursion on the subtree $T_{u}$. Finding the largest subtree or the next subtree in the size ordering is done in linear time (at most $cn$ for some constant $c$). Since \textsc{solve}($v_i,0)$ succeeded, for all $i\in [k]$ there is at most one $p_i$ for which \textsc{solve}($v_i,p_i)$ fails. We have to run \textsc{solve}$(v_i,p')$ for $p'>p\geq 1$ only when \textsc{solve}$(v_j,p)$ fails for some $j>i$. This can happen at most $k+1-i$ times. We find that
\begin{align*}
    T(n)&\leq  T(x_{k+1}) + cn + \sum_{i=1}^{k}T(x_i) + \sum_{i=1}^{k} (k+2-i)(T(x_i)+cn)\\
    &\leq T(x_{k+1}) + \frac{c}{2}n^3 + \sum_{i=1}^{k}(k+3-i)T(x_i).
\end{align*}
    Let $y= x_1+\dots+x_k$.
    Note that $x_i\leq \frac{y}{k+1-i}$: $x_{i}$ is the smallest among the $k+1-i$ elements $x_i,x_{i+1}, \dots,x_{k}$, which sum up to at most $y$. So for $T(x)=x^\alpha$, we find
    \[
    \sum_{i=1}^{k}(k+3-i)T(x_i)\leq 
    \sum_{i=0}^{k-1}(i+3)\left(\frac{y}{i+1}\right)^{\alpha}\leq
    3y^{\alpha}\sum_{i=1}^{\infty} \frac{1}{i^{\alpha-1}}.
    \]
    For $\alpha=4$, $\sum_{i=1}^{\infty} \frac{1}{i^{\alpha-1}} \leq \frac{4}{3}$.
    So with $T(x)=Cx^4$ for some constant $C\geq c$, the time bound of $T(n)$ holds for all sufficiently large $n$, since 
    \[
    T(n)\leq \max_{0\leq y\leq n/2}
    C(n-1-y)^{4} + \frac{c}{2}n^3 + 4Cy^{4}
    \]
    Note that $h:y \mapsto C(n-1-y)^4 + \frac{c}{2}n^3 + 4Cy^4$ is a convex function so its maximum value on $[0,n/2]$ is reached at 0 or at $n/2$. Hence, since
    $$h(0) = C(n-1)^4+\frac{c}{2}n^3 \leq Cn^4$$
    and
    $$h(n/2)=C(n/2-1)^4 + \frac{c}{2}n^3 + 4C(n/2)^4 \leq \frac{5}{16}Cn^4 + \frac{c}{2}n^3 \leq Cn^4$$
    we conclude that for all sufficiently large $n$:
    $$T(n)\leq Cn^4$$ \qedhere
\end{proof}
The two claims prove the theorem.
\end{proof}

\section{Additional details for Section \ref{subsec:alg}}
\label{app:c}
We give a more detailed description of the main algorithm. 
We indicate in blue whenever we make a \textcolor{blue}{recursive call} and in red when we make a \textcolor{red}{same-depth call}, to aid the reader to keep track of what the overhead information is that is stored before doing the call.

We also point out one more technical detail: we allow the algorithm to have a non-existing colour restriction of the form $(j,p_j)$ for $p_j>|L_j(v')|$. If the algorithm discovers this is the case, it will also reject. Note that if $v$ may receive many colours itself, it may never look at the given colour constraints and still accept; the algorithm will then reject either at a different child $v_2$ of $v'$ who does care about the colour constraint, or $L_j(v')=L(v')$ and $v'$ will reject itself in step 2 at some point.
\begin{enumerate}
    \item[0.] Let $v$ be the vertex pointed at by $\pos_\r$.
    Suppose first that $v$ is a leaf. We return \texttt{false} if $|L(v)|=0$ and \texttt{true} if $|L(v)|\geq 2$. \\
    The remaining case is $|L(v)|=1$. If $\colres_\r$ gives no colour restrictions, return $\texttt{true}$. Otherwise, let $v'$ denote the critical parent of $v$ and let $p_{j'}'$ be the position of the colour in $L_{j'}(v')$, where $(j',p_{j'})$ has been encoded in $\colres_\r$. We compute $p_\M'$ from $p_{j'}$ as is done in step 3. We check if this colour equals the colour in $L(v)$, and return \texttt{true} or \texttt{false} accordingly. 
    
    So we assume now that $v$ is not a leaf, so it has some heavy child $h$. We maintain that $\colres_\r$ encodes at most one colour $c'$ that $v$ must avoid. We go to 1, which will eventually lead us to one of the following:
    \begin{enumerate}
    \item[(rej)] There is no list colouring of $G$ avoiding $c'$ for $v$. We return \texttt{false}.
    \item[(nc)] The vertex $v$ can get two colours (unequal to $c'$) in list colourings of $G$. In this case, we say $v$ is non-critical. We update $\pos_\r$ such that we point at $h$, update $\colres_\r$ to `no colour constraint' and go to the start of 0 again.
    \item[(cr)] There is a unique colour $c\neq c'$ that $v$ can get in a list colouring of $G$. In this case, $L_0(v)\subseteq \{c,c'\}$ and so $p_0(c,v)$ can be represented in a single bit. We update $\pos_\r$ such that we point at $h$, update $\colres_\r$ to $(0,p_0)$ and go to the start of 0 again.
\end{enumerate}
    \item We first check that all subtrees can be coloured if we do not have any colour restrictions.
    \begin{enumerate}[(i)]
        \item Set $j=\M$ (written in unary).
        \item Set $L$ to be the number of children of $v$ whose subtree falls in the $j$th bracket and set $\ell=1$.
        \item If $\ell\leq L$, we check as follows that $T_u$ can be list coloured for the $\ell$th child $u$ of the $j$th bracket of $v$. 
        We save the relevant information of our current state, namely $($`$1(iii)$'$,\pos_\r,\colres_\r,\ell,L,j)$, to $\auxi_\r$.
        We then update $\pos_{\r}=(\pos_\r,j,\ell)$ and $\colres_r$ to `no colour constraint'. We do a \textcolor{blue}{recursive call} on \texttt{process} and then reset $\ell,j,L,\pos_\r,\colres_r$ again using $\auxi_\r$ (and remove them again from $\auxi_\r$). \\
        If the outcome was \texttt{true}, increase $\ell$ by one and repeat 1(iii). \\
        If the outcome was \texttt{false}, we go to (rej): if $T_u$ cannot be list coloured if there are no constraints, then $T$ cannot be list coloured.
        \item  At some point $\ell>L$. If $j>0$, we decrease $j$ and go to 1(ii). \\
        If $j=0$, we checked all brackets. If $|L(v)|>d(v)+2$, then $v$ is non-critical so we move to (nc). 
    \end{enumerate}
\item We start with $p_0=1$ and $j=0$ (again, $j$ is written in unary). Suppose that we wish to compute $p_{j+1}$ from $p_{j}$ for some $j\in [0,\M-1]$. 
\begin{enumerate}[(i)]
    \item Initialise $\texttt{curr}_j=1$. This represents a position in $L_j(v)$ (giving the number of `successes').\\
Initialise $\texttt{prev}_j=1$. This represents a position in $L_{j+1}(v)$ (giving the number of `tries').  
    \item Check if the $\texttt{prev}_j$th colour of $L_{j+1}(v)$ works for the trees in the $j$th bracket. 
    \begin{enumerate}[(a)]
        \item Set $L$ to be the number of children of $v$ whose subtree falls in the $j$th bracket and set $\ell=1$.
        \item If $\ell\leq L$, we check as follows that $T_u$ can be list coloured for the $\ell$th child $u$ of the $j$th bracket of $v$, assuming that $v$ gets the $\texttt{prev}_j$th colour from $L_j(v)$.\\
        We append $($`$2(ii)(b)$'$,p_j,\texttt{curr}_j,\texttt{prev}_j,\pos_\r,\colres_\r,\ell,L,j)$ to $\auxi_\r$, set $\pos_{\r}=(\pos_\r,j,\ell)$ and $\colres_r=\texttt{prev}_j$, do a \textcolor{blue}{recursive call} on \texttt{process} and then reset $\texttt{curr}_j,\texttt{prev}_j,p_j,\ell,j,L,\pos_\r,\colres_r$ again using $\auxi_\r$ (and also remove them again from $\auxi_\r$). \\
        If the outcome was \texttt{true}, we increase $\ell$ by one and repeat this step (b). If the outcome was \texttt{false}, we go to (iv).
        \item  Otherwise, $\ell>L$ and all steps succeeded. If $\texttt{curr}_j<p_{j}$, we increase both $\texttt{curr}_j$ and $\texttt{prev}_j$ and go to step (ii). \\
    If $\texttt{curr}_j=p_{j}$, then we have successfully computed $p_{j+1}$ as $\texttt{prev}_j$. 
    \end{enumerate}
    \item Suppose that one of the checks above failed. If $\texttt{prev}_j\geq 2^{j+1}+3$, then $|L_{j+1}(v)|\geq 2^{j+1}+3$ and we go to (nc). Otherwise, we increase $\texttt{prev}_j$ and go to step (ii).
\end{enumerate}
We compute $p_\M$ from $p_{j+1}$ via $p_{j+2},\dots,p_{\M-1}$.
Once we have computed $p_\M$ from $p_0=1$, we check that $p_\M\leq |L(v)|$; if it is not, we go to (rej).
\item Otherwise, we have verified that $|L_0(v)|\geq 1$. We repeat this for $p_0=2$ and $p_0=3$, although now we do not reject if we cannot find a corresponding $p_\M$. 
    \begin{itemize}
        \item If $|L_0(v)|\geq 3$, then we move to (nc). If there were no colour constraints (in $\colres_\r$) and $|L_0(v)|=2$, then we also move (nc).
        \item If there were no colour constraints, we move to (cr) where we will set $\colres$ to $j=0,p_0=1$.
    \end{itemize}
Suppose a colour $c'$ has been encoded via the position $p_{j'}'$ of $c'$ in $L_j(v')$ for some $j'\in [0,\M']$ (where $\M'=\M_{v'}$). 

We recompute $p_{\M'}'$ from $p_j'$ in a similar manner to the above. Recall that we must first recompute $p_{j'}$ from $p_0$ and then simultaneously recompute $p_{j'+1}$ from $p_{j'}$ and $p'_{j'+1}$ from $p'_{j'}$ until we computed $p_\M$ and $p_{\M'}'$. We can then check whether the $p_\M$th colour of $L(v)$ equals $c'$, the $p_{\M'}'$th colour of $L(v')$. 

We spell our now how we compute $p_{x+1}'$ from $p_{x}'$ for the parent $v'$ of $v$.
\begin{enumerate}[(i)]
    \item Initialise $\texttt{curr}_{x}=1$ and $\texttt{prev}_{x}=1$.
    \item Set $\ell=1$ and $L$ to be the number of children of $v'$ whose subtree falls in the $x$th bracket.
    \item If $\ell\leq L$, then we check the $\ell$th tree of the $x$th bracket. We append  $($`$3(iii)$'$,\alpha, p_{x+1}, p_{x}', \pos_\r, \colres_\r, \ell, L, x, \texttt{prev}_{x},$ $\texttt{curr}_{x},p_0)$ to $\auxi_r$, point $\pos_\r$ to the correct non-heavy child of $v'$ and set $\colres_r=(x+1,\texttt{prev}_{x})$.
    \begin{itemize}
        \item Suppose $v$ is the heavy child of $v'$.
        We do a \textcolor{blue}{recursive call} on \texttt{process}, similar to before. 
        \item Suppose $v$ is a non-heavy child of $v'$. We do a \textcolor{red}{same depth call} on \texttt{process}: we call it without increasing $\r$ and do not decrease $\r$ at the end of it.
    \end{itemize}
    We reset everything again using $\auxi_r$, also cleaning up in $\auxi_r$ again.\\
    If we received \texttt{false}, we increase $\texttt{prev}_x$ and repeat step (ii).\\
    If we received \texttt{true}, we increase $\ell$ by one and repeat step (iii).
    \item If $\ell>L$, we have checked all trees in the $j$th bracket. \\
    If $\texttt{curr}_x<p_{x}'$
    we increase both $\texttt{curr}_x$ and $\texttt{prev}_x$ and repeat step (ii).\\
    If $\texttt{curr}_x=p_{x}$, we successfully computed $p_{x+1}'=\texttt{prev}_x$.
    \end{enumerate}
    If $p_{\M'}>|L(v')|$, we go to (rej).\\
    If $|L_0(v)|=1$, we go to (rej) if the $p_{\M'}$th colour of $L(v')$ is equal to the $p_\M$th colour of $L(v)$, and go to (cr) otherwise. \\
    If $|L_0(v)|=2$, and the colour are equal and $p_0=1$, we go to (cr) with $j=0,p_0=2$ (since $v$ must get the $p_0$th colour of $L(v)$). If the colours are equal and $p_0=2$, we go to (cr) with $j=0,p_0=1$. \\
    If $p_0=2$ and the colours are not equal, we go to (nc). If $p_0=1$ and the colours are not equal, we repeat the thing above for $p_0=2$, first removing $p_{\M}$ and $p_{\M'}'$ again to clean up the space.
\end{enumerate}

\end{document}